\documentclass{article}

\usepackage{arxiv}

\usepackage[utf8]{inputenc} % allow utf-8 input
\usepackage[T1]{fontenc}    % use 8-bit T1 fonts
\usepackage{hyperref}       % hyperlinks
\usepackage{url}            % simple URL typesetting
\usepackage{booktabs}       % professional-quality tables
\usepackage{amsfonts}       % blackboard math symbols
\usepackage{nicefrac}       % compact symbols for 1/2, etc.
\usepackage{microtype}      % microtypography
\usepackage{lipsum}		% Can be removed after putting your text content
\usepackage{graphicx}
\usepackage{natbib}
\usepackage{doi}

\usepackage{amsmath}
\usepackage{amsthm}
\usepackage{optidef}
\usepackage{algorithm}
\usepackage{algorithmic}
\usepackage{multirow}
\usepackage{pgfplots}
\usepackage{subcaption}

\pgfplotsset{compat=1.18}
\pgfplotsset{
  cycle list={
    {red,   solid, mark=*, mark size=0.8pt, mark options={fill=red}},
    {blue,  solid, mark=*, mark size=0.8pt, mark options={fill=blue}},
    {green, solid, mark=*, mark size=0.8pt, mark options={fill=green}},
    {orange,solid, mark=*, mark size=0.8pt, mark options={fill=orange}},
    {purple,solid, mark=*, mark size=0.8pt, mark options={fill=purple}},
    {olive, solid, mark=*, mark size=0.8pt, mark options={fill=olive}},
  }
}

\theoremstyle{plain}
\newtheorem{thm}{Theorem}
\theoremstyle{definition}
\newtheorem{dfn}{Definition}

\AtBeginDocument{%
  }

\title{Kd-tree Based Wasserstein Distance Approximation\\for High-Dimensional Data}

\author{{\hspace{1mm}Kanata Teshigawara}\\
	Institute of Science Tokyo\\
	\texttt{teshigawara.k.46c4@m.isct.ac.jp} \\
	\And
    {\hspace{1mm}Keisho	Oh} \\
	Indeed Recruit Technologies Co., Ltd.\\
	\texttt{oh\_keisho@r.recruit.co.jp} \\
	\And
    {\hspace{1mm}Ken Kobayashi} \\
	Institute of Science Tokyo \\
	\texttt{kobayashi.k@iee.eng.isct.ac.jp} \\
	\And
    {\hspace{1mm}Kazuhide	Nakata} \\
	Institute of Science Tokyo \\
	\texttt{nakata.k.ac@m.titech.ac.jp} \\
}

\hypersetup{
pdftitle={A template for the arxiv style},
pdfsubject={q-bio.NC, q-bio.QM},
pdfauthor={David S.~Hippocampus, Elias D.~Striatum},
pdfkeywords={First keyword, Second keyword, More},
}

\begin{document}
\maketitle

\begin{abstract}
The Wasserstein distance is a discrepancy measure between probability distributions, defined by an optimal transport problem. 
It has been used for various tasks such as retrieving similar items in high-dimensional images or text data. 
In retrieval applications, however, the Wasserstein distance is calculated repeatedly, and its cubic time complexity with respect to input size renders it unsuitable for large-scale datasets. 
Recently, tree-based approximation methods have been proposed to address this bottleneck. 
For example, the Flowtree algorithm computes transport on a quadtree and evaluates cost using the ground metric, and clustering-tree approaches have been reported to achieve high accuracy. 
However, these existing trees often incur significant construction time for preprocessing, and crucially, standard quadtrees cannot grow deep enough in high-dimensional spaces, resulting in poor approximation accuracy. 
In this paper, we propose \emph{kd-Flowtree}, a kd-tree-based Wasserstein distance approximation method that uses a kd-tree for data embedding. 
Since kd-trees can grow sufficiently deep and adaptively even in high-dimensional cases, \emph{kd-Flowtree} is capable of maintaining good approximation accuracy for such cases. 
In addition, kd-trees can be constructed quickly than quadtrees, which contributes to reducing the computation time required for nearest neighbor search, including preprocessing. 
We provide a probabilistic upper bound on the nearest-neighbor search accuracy of kd-Flowtree, and show that this bound is independent of the dataset size.
In the numerical experiments, we demonstrated that kd-Flowtree outperformed the existing Wasserstein distance approximation methods for retrieval tasks with real-world data.
\end{abstract}

\section{Introduction}
Optimal Transport (OT) is an optimization problem that computes the minimum cost required to transform one probability distribution into another.  
Under appropriate assumptions, the resulting minimal total transport cost satisfies the axioms of a distance metric and is known as the Wasserstein distance.  
In recent years, the Wasserstein distance has attracted significant interest in data science as a tool for handling complex data such as point clouds and probability measures.  
For example, the Wasserstein GAN (WGAN), which utilizes the Wasserstein distance as its loss function, was proposed by Arjovsky et al.~\cite{arjovsky2017wasserstein}.  
In domain adaptation, methods that compare the joint distributions of features and labels between source and target domains using the Wasserstein distance have outperformed conventional approaches on image recognition tasks~\cite{courty2017joint, damodaran2018deepjdot}. 
Furthermore, in bioinformatics, the Wasserstein distance has been applied to accelerate and enhance sensitivity in detecting distributional differences\cite{schefzik2021fast}, to evaluate intercellular similarity, to correct batch effects, to augment data, and to measure variability in gene expression\cite{huizing2022optimal, tabak2020correcting, zhang2022ideas, li2023mdwgan}.  

In this way, the Wasserstein distance has been applied widely, and one of its most prominent applications is the retrieval of similar items in data such as images and natural language text.  
For document retrieval, the Word Mover’s Distance—defined as the optimal transport cost between distributions in a word embedding space—has been shown to evaluate semantic similarity with high accuracy even for documents with no overlapping vocabulary\cite{kusner2015word}.
As a task‐specific optimization of this idea, the Supervised Word Mover’s Distance introduces pseudo‐label–based distance metric learning to further improve classification and retrieval accuracy~\cite{huang2016supervised}.  
In image data, it has been proposed to incorporate a Wasserstein loss into the ground distance used for measuring semantic similarity of output labels~\cite{frogner2015learning}.  
Furthermore, DeepEMD treats local regions on deep feature maps as masses for optimal matching and achieves high performance in few‐shot image classification tasks~\cite{zhang2020deepemd}, and DeepFace‐EMD performs patch‐wise EMD re‐ranking to enable robust face retrieval under transformations such as masking and rotation~\cite{phan2022deepface}.
These results demonstrate that the Wasserstein distance is an extremely effective method for retrieving similar items in both image and text data. 

In similar‐item retrieval, one must perform a nearest‐neighbor search to find the database item that minimizes the Wasserstein distance to a query.  
This process requires repeated computation of the Wasserstein distance.  
However, it is well known that the exact computation of the Wasserstein distance generally requires $O(s^3 \log s)$ time with respect to the support size $s$ of the distributions~\cite{pele2009fast}.  
Therefore, searching over large‐scale data with $10^5$ or $10^6$ candidate items is impractical, making the use of appropriate approximation methods indispensable.  
A widely used approximation is the Sinkhorn distance, but its computation takes $O(s^2)$ time in $s$~\cite{cuturi2013sinkhorn}, which remains insufficient for large‐scale data.  

Recent studies have focused on embedding data points from the ground space into a tree and utilizing the Wasserstein distance defined on that tree~\cite{le2019tree, takezawa2021supervised, le2024optimal,chen2024learning}.  
The tree‐Wasserstein distance (TWD) is computed by summing the weighted mass transported across each edge and can be calculated in linear time~\cite{le2019tree}.  
Moreover, TWD provides both upper and lower bounds on the Wasserstein distance in the original ground space, and it has been shown that the approximation error can be theoretically controlled by the embedding method~\cite{le2019tree}.  
In other words, with an appropriate embedding, TWD can serve as an effective approximation technique.  
Related to this, Flowtree was proposed to approximate the Wasserstein distance by computing an optimal transport on the embedded tree and then evaluating it using the ground space cost~\cite{backurs2020scalable}.  
This algorithm also runs in linear time and has the property that its search accuracy does not depend on the dataset size.  
The original Flowtree paper uses a quadtree embedding and demonstrates strong theoretical performance.  
However, in practice, it often falls short, leaving room for more suitable embedding methods.  

As an alternative approach, methods that embed data points into a one‐dimensional space (a line) and use a greedy algorithm for approximation have also been proposed.  
Specifically, the support of each distribution is projected onto a line using clustering or random projections, and then points that are close in position on the line are greedily matched to approximate the Wasserstein distance.  
This greedy algorithm itself runs in linear time and achieves relatively high approximation accuracy.  
However, significant information loss during the embedding process can cause accuracy to vary depending on the data.  
Moreover, when constructing embeddings using a clustering tree, the preprocessing step usually requires a substantial amount of time.  

In this paper, we propose \emph{kd-Flowtree}, a method that constructs embeddings with a kd‐tree instead of a quadtree and then applies the Flowtree algorithm for approximation.  
The Flowtree algorithm’s approximation accuracy improves as the depth of the constructed tree increases.  
When processing text data, the ground space is often high‐dimensional, but quadtrees frequently cannot grow sufficiently deep in such spaces, whereas kd‐trees can achieve adequate depth.  
Consequently, kd‐Flowtree retains Flowtree’s linear‐time computation while substantially enhancing nearest‐neighbor search performance.  
Furthermore, kd‐trees can be built more quickly than quadtrees, reducing the total time required for nearest‐neighbor search, including preprocessing.  
We also prove that kd‐Flowtree’s search accuracy remains independent of the dataset size.   

This paper is organized as follows.  
First, we present the formulation of OT and the details of the approximate computation methods for the Wasserstein distance introduced above (Section 2).  
Next, we describe the proposed \emph{kd-Flowtree}, including the specific kd‐tree construction method, its computational time, and the theoretical aspects of its approximation accuracy (Section 3).  
Finally, we conduct numerical experiments on text data to demonstrate the effectiveness of the proposed method for similarity retrieval tasks (Section 4).  

Overall, the contributions of this paper are as follows.  
\begin{itemize}
  \item We propose \emph{kd-Flowtree}, which constructs embeddings using a kd‐tree instead of a quadtree and then computes approximation costs with the same algorithm as Flowtree. The computational time remains comparable, but nearest‐neighbor search accuracy significantly surpasses that of Flowtree. This improvement stems from the fact that kd‐trees can grow deeper than quadtrees in high‐dimensional spaces. Moreover, kd‐trees can be built relatively quickly, keeping preprocessing efficient even in high‐dimensional settings. Clustering trees and quadtrees often incur substantial construction time, so even if the downstream algorithm is fast, preprocessing can be time‐consuming; our proposed method overcomes this drawback.  
  \item We experimentally demonstrate that, for both Flowtree and kd-Flowtree, approximation accuracy in nearest‐neighbor search depends on the depth of the constructed tree, with deeper trees yielding better performance. In high‐dimensional spaces, quadtrees struggle to achieve sufficient depth, whereas kd‐trees do so readily. Consequently, kd-Flowtree exhibits superior search performance, especially for text data where the ground space dimension is typically large.
  \item We prove that, as with Flowtree, the search accuracy of kd-Flowtree remains independent of the dataset size. Although the theoretical guarantees for kd‐Flowtree depend on dimensionality—unlike the quadtree‐based Flowtree—this dependency is minor in practice, and the choice of embedding leads to substantial accuracy improvements.
\end{itemize}

\section{Related Work}
In this section, we present the definition of OT and review previously proposed approximation methods.

\subsection{Definition and Computational Complexity of OT}
Optimal Transport (OT) is an optimization problem that computes the minimum transport cost required to transform one probability distribution into another~\cite{rubner2000earth}. 
In this paper, we consider OT for discrete probability distributions (point sets).  
\begin{dfn}[Optimal transport between point sets~\cite{villani2008optimal}]
Let $\mu,\nu$ be probability distributions on $\mathbb{R}^D$, represented by probability vectors $\boldsymbol{\mu}\in\mathbb{R}_+^n$ and $\boldsymbol{\nu}\in\mathbb{R}_+^m$ satisfying $\sum_{i=1}^n\mu_i = \sum_{j=1}^m\nu_j = 1$.  
Given a cost matrix $\boldsymbol{C}\in\mathbb{R}_+^{n\times m}$, the optimal transport problem between $\mu$ and $\nu$ is defined as
\begin{gather*}
\mathrm{OT}(\mu, \nu)
= \min_{\boldsymbol{P}\in\mathcal{U}(\boldsymbol{\mu}, \boldsymbol{\nu})} \langle \boldsymbol{C}, \boldsymbol{P}\rangle, \\
\text{where}\quad
\mathcal{U}(\boldsymbol{\mu}, \boldsymbol{\nu})
= \left\{\boldsymbol{P}\in\mathbb{R}_+^{n\times m}\,|\,\boldsymbol{P}\boldsymbol{1}_m=\boldsymbol{\mu},\,\boldsymbol{P}^\top\boldsymbol{1}_n=\boldsymbol{\nu}\right\}.
\end{gather*}
\end{dfn}
When the cost matrix $\boldsymbol{C}$ is defined by $C_{ij} = d(\boldsymbol{x}_i,\boldsymbol{y}_j)^p$ for a distance metric $d$, $\mathrm{OT}(\mu,\nu)$ satisfies the axioms of a distance. We then refer to $\mathrm{OT}(\mu,\nu)^{1/p}$ as the $p$-Wasserstein distance.

\begin{dfn}[$p$-Wasserstein distance~\cite{villani2008optimal}]
Let $p>0$.  
Consider the metric space $(\mathbb{R}^D,d)$ and two probability distributions $\mu,\nu$ with supports $Supp(\mu)=\{\boldsymbol{x}_1,\dots,\boldsymbol{x}_n\}$ and $Supp(\nu)=\{\boldsymbol{y}_1,\dots,\boldsymbol{y}_m\}$, represented by vectors $\boldsymbol{\mu}\in\mathbb{R}_+^n$ and $\boldsymbol{\nu}\in\mathbb{R}_+^m$.
Define the cost matrix $\boldsymbol{C}\in\mathbb{R}_+^{n\times m}$ by $C_{ij} = d(\boldsymbol{x}_i,\boldsymbol{y}_j)^p$.  Then the $p$-Wasserstein distance between $\mu$ and $\nu$ is
\begin{equation*}
    W_p(\mu,\nu) \;=\; \mathrm{OT}(\mu,\nu)^{1/p}.
\end{equation*}
\end{dfn}

In this work, we focus on the case $p=1$, referred to simply as the Wasserstein distance and denoted $W_1(\mu,\nu)$. Unless otherwise stated, the distance $d$ is the $\ell_1$ metric.

The computational complexity of the Wasserstein distance is known to be $O(N^3 \log N)$ with respect to the total number of support points $N = n + m$, making repeated calculations on large‐scale data impractical~\cite{pele2009fast}.  
A widely used approximation is the Sinkhorn distance, which is obtained by adding a regularization term to the objective function of the OT problem and solving the modified optimization~\cite{cuturi2013sinkhorn}.  
However, its computation still requires $O(N^2)$ time in $N$~\cite{cuturi2013sinkhorn}, which remains insufficient for large‐scale datasets.  

\subsection{Approximation via Tree Embedding}
As a faster approximation than the Sinkhorn distance, approaches based on embedding into a tree are known.  
In this approach, a probability distribution in a high‐dimensional ground space is first embedded into a tree whose leaves correspond to the data points.  
On the tree, the cost of transporting mass across each edge can be computed simply as a weighted sum, and OT can be expressed as follows.

\begin{dfn}[Tree‐based 1‐Wasserstein Distance~\cite{le2019tree}]
Let $T = (X, V, E)$ be a tree with node set $V$, leaf set $X \subset V$, and edge set $E \subset V \times V$.  
Let $w: E \to \mathbb{R}_+$ be the weight function assigning a positive weight to each edge of $T$, and define the tree distance $t(u,v)$ between two nodes $u,v$ as the sum of weights along the unique path between them in $T$.  
For two weight distributions $\mu,\nu$ on $V$, the tree‐based 1‐Wasserstein distance (TWD) $W_1^t(\mu,\nu)$ is given by
\begin{equation*}
    W_1^t(\mu,\nu) \;=\; \sum_{v\in V} t\left(v, p(v)\right)\,\left|\mu(v)-\nu(v)\right|,    
\end{equation*}

where $p(v)$ denotes the parent of node $v$, and $\mu(v),\nu(v)$ are the weights of distributions $\mu,\nu$ at node $v$, satisfying $\sum_{v\in V}\mu(v)=\sum_{v\in V}\nu(v)=1$.
\end{dfn}

Since the above expression aggregates weights only along the edges of the tree, it can be computed in $O(|V|)$ time, proportional to the number of nodes.  
The resulting TWD has been shown theoretically to act as both an upper and lower bound on the true Wasserstein distance in the original Euclidean space, and by designing the embedding into the tree appropriately, one can control the approximation error~\cite{le2019tree}.  

Flowtree significantly improves the traditional TWD by making small modifications to the tree‐based approach~\cite{backurs2020scalable}.  
Instead of directly computing distances on the tree, Flowtree first embeds the ground space using a quadtree and then finds an optimal transport matching $\widetilde{\boldsymbol{P}}$ based on the tree distance $t(v,u)$.  
It then uses this matching to estimate the true Wasserstein distance under the metric $d$.

\begin{dfn}[Flowtree\cite{backurs2020scalable}]
Embed two distributions $\mu,\nu$ on the metric space $(\mathbb{R}^D,d)$ into a quadtree $T=(X,V,E)$.  
Let their supports in the tree be $\mathrm{Supp}_t(\mu)=\{u_1,\dots,u_n\}\subset V$ and $\mathrm{Supp}_t(\nu)=\{v_1,\dots,v_m\}\subset V$.  
Define the depth $l$ of nodes relative to the support diameter $\Phi$, starting from the root as $\log\Phi+1,\log\Phi,\log\Phi-1,\dots$, and assign each edge connecting a node at depth $l$ to a node at depth $l+1$ a weight of $2^l$.  
Then the optimal transport matching $\widetilde{\boldsymbol{P}}$ is given by
\begin{equation*}
    \widetilde{\boldsymbol{P}}
    = \arg\min_{\boldsymbol{P}\in\mathbb{R}^{n\times m}}
      \sum_{i=1}^n \sum_{j=1}^m P_{ij}\,t(u_i,v_j).    
\end{equation*}
The Flowtree‐based approximate Wasserstein distance $\widetilde{W}_1(\mu,\nu)$ is then
\begin{equation*}
    \widetilde{W}_1(\mu,\nu)
    = \sum_{i=1}^n \sum_{j=1}^m \widetilde{P}_{ij}\,d(\boldsymbol{x}_i,\boldsymbol{y}_j).
\end{equation*}
\end{dfn}

The above sequence of computations can be implemented as shown in Algorithm \ref{alg:flow}.
\begin{figure}[ht]
\begin{algorithm}[H]
  \caption{Flowtree}
  \label{alg:flow}
  \begin{algorithmic}[1]
    \REQUIRE A tree $T$ with root $r$ and height $h$, and distributions $\mu,\nu$
    \STATE Initialize $\widetilde{P}_{ij} \gets 0$ for all $(i,j)$
    \FOR{each node $q$ of $T$, processed from leaves upward}
      \IF{$\exists\,i\text{ such that }q=u_i$}
        \STATE $U_\mu\gets U_\mu \cap \{\mu_i\}$
      \ENDIF
      \IF{$\exists\,j\text{ such that }q=u_j$}
        \STATE $U_\nu\gets U_\nu \cap \{\nu_j\}$
      \ENDIF
      \WHILE{there exist $\mu_i\in U_\mu$ and $\nu_j\in U_\nu$}
        \IF{$\mu_i \geq \nu_j$}
          \STATE $\eta \gets \nu_i$
          \STATE $U_\nu\gets U_\nu \setminus \{\nu_j\}$
        \ELSE
          \STATE $\eta \gets \mu_i$
          \STATE $U_\mu\gets U_\mu \setminus \{\mu_i\}$
        \ENDIF
      \STATE $\widetilde{P}_{ij} \gets \widetilde{P}_{ij} + \eta$
      \STATE $\mu_i \gets \mu_i - \eta$
      \STATE $\nu_j \gets \nu_j - \eta$
      \ENDWHILE
    \ENDFOR
    \RETURN $\widetilde{\boldsymbol{P}}$
  \end{algorithmic}
\end{algorithm}
\end{figure}
Including distance computation, the time complexity for calculating $\widetilde{W}_1(\mu,\nu)$—excluding tree construction—is $O\left(s\,(D + h)\right)$, where $s = \max(n,m)$, $D$ is the dimension, and $h$ is the tree height.  
Moreover, in the quadtree‐based Flowtree, the approximation accuracy of nearest‐neighbor search is independent of the dataset size.  
However, when this method is applied to real datasets—especially those in high‐dimensional spaces—its performance is not always satisfactory.  
This is because quadtrees cannot grow sufficiently deep in high dimensions.  
Thus, there is a need for methods that maintain strong performance even in high‐dimensional spaces.  

\subsection{One-Dimensional Embedding}
As an alternative approach, methods that project the ground space onto a one-dimensional line and compute an approximate transport cost using a greedy algorithm have been proposed.  

First, when the ground space is one-dimensional, OT admits an exact solution via a greedy procedure~\cite{rachev1998mass}. 
Specifically, one sorts the support points of both distributions in ascending order along the line and then repeatedly transports mass from one distribution to the other in that order until all allocations are complete.  

For faster approximation, two greedy algorithms—$k$-Greedy and 1D-ICT—have been introduced~\cite{houry2024fast}. Both algorithms pair nearby support points on the line and transport mass within each pair. 
In $k$-Greedy, increasing $k$ improves accuracy, although even relatively small values of $k$ yield strong performance. 1D-ICT is somewhat less accurate than $k$-Greedy but offers faster computation.  
Two projection schemes are used to embed data into one dimension: one based on a clustering tree and another based on random projections. 
While the greedy matching itself runs very quickly, constructing the clustering tree can incur substantial preprocessing time. 
Moreover, significant information loss during the embedding process can cause the approximation accuracy to vary depending on the dataset.  

\section{Proposed Method}
In this section, we describe the proposed method, \emph{kd-Flowtree}.  
First, we compare the kd‐tree construction procedure, its depth, and computational time against those of quadtrees.  
Next, we detail the computational time and approximation accuracy of the proposed method.

\subsection{kd‐Tree Construction}
A kd‐tree is a data structure that recursively partitions space by alternating splitting dimensions at each node and dividing at the median coordinate value.  
In this paper, we build upon the static kd‐tree construction algorithm of Bentley~\cite{bentley1975multidimensional}, augmenting it with a random shift at each split.
This random shift is analogous to that applied by Backurs et al.~\cite{backurs2020scalable} to quadtrees and is essential for providing theoretical guarantees.  

Consider constructing a kd‐tree on $N$ points in $\mathbb{R}^D$.  
First, select a splitting axis $k$ uniformly at random (with probability $1/D$) and compute the median $m$ of the coordinates along that axis.  
Next, choose a constant $\eta\in[0,0.5)$ and generate a uniform random shift $\delta\in[-\eta L_k,\,\eta L_k]$, where $L_k$ is the width of the cell along axis $k$.  Set $m' = m + \delta$, and partition the point set into  
\begin{equation*}
  X_L = \{\boldsymbol{x}\in\mathbb{R}^D \mid x_k < m'\}, \quad
  X_R = \{\boldsymbol{x}\in\mathbb{R}^D \mid x_k \ge m'\}.
\end{equation*}  
Apply this procedure recursively to each subset until each cell contains at most one point.  The resulting tree has average depth $O(\log N)$ and can be built in $O(N\log N)$ time~\cite{bentley1975multidimensional}, a construction cost that does not depend on the dimension $D$.  

The quadtree used in Flowtree is a data structure constructed by recursively partitioning the space into $2^D$ congruent hypercubes along each coordinate axis.  
Originally introduced for two‐dimensional data~\cite{finkel1974quad}, it has since been applied to high‐dimensional datasets~\cite{faloutsos2002analysis, kratochvil2020generalized}.  
To build the tree, one first determines the smallest hypercube containing all $N$ points.  
If a hypercube contains more than one point, it is split along each axis at its midpoint, yielding $2^D$ sub‐hypercubes.  
For each sub‐hypercube, if it contains more than one point, the same subdivision procedure is applied recursively; regions with exactly one point are not subdivided further.  
In this case, the resulting quadtree has depth $O\left(\tfrac{\log N}{D}\right)$ and can be constructed in $O(ND\log N)$ time~\cite{bern1993approximate}.  

\subsection{Advantages of kd-Flowtree}
The difference between the proposed kd-Flowtree and the prior Flowtree algorithm is the replacement of the quadtree with a kd‐tree for partitioning the ground space; the procedure for computing $\widetilde{W}_1$ using Algorithm~\ref{alg:flow} remains the same. 
However, this minor modification yields significant practical performance improvements.

First, the fact that the tree depth does not depend on the dimension leads to a substantial increase in approximation accuracy. 
The depth of a quadtree is $O\bigl(\tfrac{\log N}{D}\bigr)$, which decreases as the dimension $D$ increases. For example, when $D=50$, a single split creates $2^{50}\approx10^{15}$ regions, causing most quadtrees to terminate after one split. 
However, as shown in the numerical experiments of the next chapter, achieving sufficient tree depth crucially impacts approximation accuracy, with deeper trees yielding better results. 
In contrast, kd‐trees have depth $O(\log N)$ independent of $D$, allowing them to grow deep even in high dimensions and consistently deliver high accuracy.

Moreover, because kd‐trees can be built quickly, they reduce not only the time to compute $W_1$ but also the overall preprocessing time. 
As noted earlier, kd‐tree construction runs in $O(N\log N)$ time independent of the dimension, whereas quadtree construction requires $O(ND\log N)$ time proportional to $D$. 
In particular, for natural language text—where embedding spaces often have dimensions of 50, 100, or higher~\cite{mikolov2013distributed, pennington2014glove, devlin2019bert}—quadtree construction can be time‐consuming. 
In contrast, kd‐trees complete construction relatively quickly, even in high dimensions. 
While some existing approximation methods (e.g., Flowtree or clustering‐tree approaches) achieve high accuracy and fast cost computation, they often incur substantial preprocessing costs. 
In this context, kd-Flowtree excels by offering both fast cost computation and low preprocessing time.

Note that, as mentioned in the previous subsection, the time complexity for computing $\widetilde{W}_1$ is $O\bigl(N\,(D + h)\bigr)$, where $h$ is the tree depth. 
Thus, Flowtree runs in $O\bigl(N\,(D + \tfrac{\log N}{D})\bigr)$ time, whereas kd-Flowtree runs in $O\bigl(N\,(D + \log N)\bigr)$ time, meaning kd-Flowtree may require slightly more time in very high dimensions. 
However, as our numerical experiments confirm, this difference is negligible in practice.   

\subsection{Accuracy of Nearest‐Neighbor Search with kd‐Flowtree}
In this section, we present accuracy guarantees for nearest‐neighbor search using kd‐Flowtree. 
We first review the guarantee for Flowtree and then prove the corresponding result for kd‐trees.

For Flowtree, the following theorem provides a bound on search accuracy under a uniform‐weight assumption using probabilistic inequalities.

\begin{thm}[Upper bound on Flowtree search accuracy under uniform weights~\cite{backurs2020scalable}]
Let $s$ be an integer, and assume that the weights in each distribution’s support are integer multiples of $1/s'$ for some $s'\le s$ (i.e., $1/s',2/s',\dots$). 
Let $\nu$ be a query distribution, $\mu^*$ its true nearest neighbor, and $\mu'$ the neighbor returned by Flowtree. 
Then, with probability at least $0.99$,
\[
   W_1(\mu',\nu) \;\le\; O(\log^2 s)\,\cdot W_1(\mu^*,\nu).
\]
\end{thm}

This theorem holds, for example, when each support weight is $1/s'$ and $s$ is the maximum support size. Notably, the search accuracy depends only on $s$ and not on the dataset size.

Next, we show that kd‐Flowtree also yields nearest‐neighbor search accuracy independent of the dataset size. 
The proof assumes that in each kd‐tree split, the chosen axis reduces the expected length of the cell along that axis by a factor of $1/2$. 
Such an assumption is common in theoretical analyses of kd‐trees~\cite{bentley1975multidimensional, friedman1977algorithm, samet1984quadtree}. 
Under this assumption, $D$ successive kd‐tree splits can be treated as equivalent to one quadtree split, and the proof then proceeds analogously to the quadtree case.

\begin{thm}[Upper bound on kd-Flowtree search accuracy under uniform weightss~\cite{backurs2020scalable}]
Let $s$ be an integer, and assume the weights in each distribution’s support are integer multiples of $1/s'$ for some $s'\le s$ (i.e., $1/s',2/s',\dots$). For a query distribution $\nu$, let $\mu^*$ be its true nearest neighbor and $\mu'$ the neighbor returned by kd‐Flowtree. 
Then, with probability at least $0.99$,
\begin{equation*}
    W_1(\mu',\nu)\leq O(D^{1-\varepsilon}\log^2 s)\cdot W_1(\mu^*,\nu),
\end{equation*}
where
\begin{equation*}
    \varepsilon = \frac{1}{\log s}
\end{equation*}
\end{thm}

\begin{proof}
    The proof consists of the following three steps:
    \begin{enumerate}
        \item Bounding the probability $p_t(\boldsymbol{x}_i,\boldsymbol{y}_j)$ that two points $\boldsymbol{x}_i,\boldsymbol{y}_j$ remain in the same cell after the $t$-th split.
        \item For $\varepsilon\in(0,1)$, defining a modified tree distance $t'(u_i,v_j)$ by weighting each edge at depth $l$ by $2^{l(1-\varepsilon)}$, and deriving both a lower bound and an upper bound on its expectation.
        \item Using the definitions of $W_1$ and $\widetilde{W}_1$ to relate these bounds to the approximation error.
    \end{enumerate}
    Here, let $l(t)$ denote the depth of the constructed tree after $t$ splits, defined as
    \begin{equation*}
        l(t) = \log\Phi - \left\lfloor\frac{t}{D}\right\rfloor.
    \end{equation*}
    The weight of the edge connecting a node at depth $l(t)$ and a node at depth $l(t+1)$ is given by
    \begin{equation*}
        2^{l(t)} = \Phi \cdot 2^{-\left\lfloor\frac{t}{D}\right\rfloor} = 2^{\log\Phi - \left\lfloor\frac{t}{D}\right\rfloor}.
    \end{equation*}
    The constant $\Phi$ is determined by setting the root node depth to $\log\Phi+1$ and the side length of the initial region (hypercube) before the first split to $2\Phi$.
    Although the algebraic manipulations follow the Flowtree proof, steps (1) and (2) exploit the specific properties of the kd-tree embedding, which distinguishes this proof from that of Flowtree. Step (3) follows identical calculations to Flowtree using the results from (1) and (2), and is therefore omitted. Constants $C_1, C_2, \dots$ and $c_a, c_b$ are used below.

    \noindent\textbf{(1) Probability bound.}
    Let $\Delta_k = |x_{ik}-y_{jk}|$, where $x_{ik}$ and $y_{jk}$ denote the $k$-th components of $\boldsymbol{x}_i$ and $\boldsymbol{y}_j$, respectively. Then,
    \begin{equation*}
        \|\boldsymbol{x}_i-\boldsymbol{y}_j\|_1 = \sum_{k=1}^D \Delta_k.
    \end{equation*}
    In the kd-tree construction, $\boldsymbol{x}_i$ and $\boldsymbol{y}_j$ are separated in a single split only if the chosen splitting axis $k$ lies between $x_{ik}$ and $y_{jk}$.
    Let $L^{(t)}_k$ be the length of the cell along axis $k$ at the $t$-th split (depth $t$). The probability of separation along axis $k$ is $\Delta_k / L^{(t)}_k$. Thus, the probability $p_t(\boldsymbol{x}_i,\boldsymbol{y}_j)$ is given by
    \begin{equation*}
        p_t(\boldsymbol{x}_i,\boldsymbol{y}_j) = \prod_{k=1}^D \left(1 - \frac1D \cdot \frac{\Delta_k}{L^{(t)}_k}\right).
    \end{equation*}
    Based on the assumption that each split reduces the expected width of the region along the chosen axis by half, $L^{(t)}_k$ satisfies
    \begin{equation*}
        \Phi \cdot 2^{-(\lfloor t/D \rfloor + 1)} \;\le\; L^{(t)}_k \;\le\; \Phi \cdot 2^{-\lfloor t/D \rfloor}.
    \end{equation*}
    Using this relation and algebraic manipulations similar to the Flowtree analysis, we obtain
    \begin{equation*}
        1 - \frac{C_{1}}{D}\cdot\frac{\|\boldsymbol{x}_i-\boldsymbol{y}_j\|_1}{2^{l(t)}}
        \;\le\;
        p_t(\boldsymbol{x}_i,\boldsymbol{y}_j)
        \;\le\;
        \exp\!\left(-\frac{C_2}{D}\cdot\frac{\|\boldsymbol{x}_i-\boldsymbol{y}_j\|_1}{2^{l(t)}}\right).
    \end{equation*}

    \noindent\textbf{(2) Tree-distance bound.}
    Set $\varepsilon = 1/\log s$ and choose a sufficiently small constant $c>0$ such that $\delta = c/s^2 > 0$.
    For $\varepsilon \in (0,1)$, define the modified tree distance $t'(u_i,v_j)$ by replacing the standard edge weights $2^l$ with $2^{l(1-\varepsilon)}$.
    Let $t_{xy}$ be the largest integer satisfying
    \begin{equation*}
        2^{l(t)} \;\le\; \frac{C_2}{D}\frac{\|\boldsymbol{x}_i-\boldsymbol{y}_j\|_1}{(\log(1/\delta))^{1/(1-\varepsilon)}}.
    \end{equation*}
    The probability $1 - p_{t_{xy}}(\boldsymbol{x}_i,\boldsymbol{y}_j)$ that $\boldsymbol{x}_i$ and $\boldsymbol{y}_j$ are separated by depth $t_{xy}$ satisfies
    \begin{equation*}
        1 - p_{t_{xy}}(\boldsymbol{x}_i,\boldsymbol{y}_j) \;\ge\; 1 - \delta^{1-\varepsilon}.
    \end{equation*}
    Since there are at most $s^2$ combinations of support pairs, applying a union bound yields
    \begin{equation*}
        \Pr[\forall(\boldsymbol{x}_i,\boldsymbol{y}_j) \text{ separated by } t_{xy}]
        \;\ge\;
        1 - s^2 \cdot \delta^{1-\varepsilon}
        \;=\;
        1 - s^2 \cdot \left(\frac{c}{s^2}\right)^{1-\varepsilon},
    \end{equation*}
    which exceeds $0.99$ for a sufficiently small $c$. Using this result and following the derivation in Flowtree, we have that with probability at least $0.99$,
    \begin{equation*}
        t'(u_i,v_j) \;\ge\; \frac{\|\boldsymbol{x}_i-\boldsymbol{y}_j\|_1^{\,1-\varepsilon}}{C_3 \cdot D^{1-\varepsilon} \log(1/\delta)}.
    \end{equation*}
    Furthermore, the upper bound on the expectation of the tree distance follows similarly as
    \begin{equation*}
        \mathbb{E}\left[t'(u_i,v_j)\right]
        \;\le\;
        C_4 \cdot \log s \; \|\boldsymbol{x}_i-\boldsymbol{y}_j\|_1^{\,1-\varepsilon}.
    \end{equation*}

    \noindent\textbf{(3) Approximation bound.}
    Finally, combining the above results with the definitions of $W_1$ and $\widetilde{W}_1$ proves the theorem.
\end{proof}

To illustrate the impact of the dimensional term $D^{1-\varepsilon}$, we present the values of the exponent $1-\varepsilon$ for various sizes of $s$ in Table~\ref{tab:epsilon_values}. Note that $\varepsilon = 1/\log_2 s$.

\begin{table}[h]
    \centering
    \caption{Values of the exponent $1-\varepsilon$ for different support sizes $s$}
    \label{tab:epsilon_values}
    \vspace{0.2cm}
    \begin{tabular}{lcccccc}
        \toprule
        $s$ & 10 & 50 & 80 & 100 & 200 & 500 \\
        \midrule
        $1 - \varepsilon$ & 0.699 & 0.823 & 0.842 & 0.849 & 0.869 & 0.888 \\
        \bottomrule
    \end{tabular}
\end{table}

The bound now contains a factor of $D$, which might appear to degrade performance in high‐dimensional spaces. 
However, in order to remove data dependence, our analysis—especially the probability bound in step (1)—is highly conservative, and as we will confirm in the next chapter, its practical impact is negligible. 
Instead, the benefit of achieving sufficient tree depth is pronounced, and our experiments show that using a kd‐tree leads to substantially improved approximation accuracy.  

\section{Experiments}
In this section, we present two numerical experiments conducted to evaluate the effectiveness of kd-Flowtree, detailing the experimental setup and discussing the results.

\subsection{Experimental Objectives and Setup}
%First, we describe the objectives and settings of the two experiments.
The first experiment evaluated nearest‐neighbor search performance on real datasets using various approximation methods. 
We used the three text datasets employed in prior work~\cite{backurs2020scalable, houry2024fast}: 20NEWS, Amazon, and BBC. 
20NEWS consists of English postings from 20 different newsgroups. 
AMAZON comprises consumer reviews across various product categories; we used data from the ``Books,'' ``DVD,'' ``Electronics,'' and ``Kitchen \& Housewares'' categories. 
BBC contains English news articles classified into five sections: Entertainment, Politics, Sports, Technology, and Business, and we used all five sections in our experiments. 
Table~\ref{tab:dataset} summarizes the dataset size and the average support size for each dataset.

\begin{table}[htbp]
    \centering
    \caption{Dataset size and average support size}
    \begin{tabular}{ccc} \hline
        Dataset & Size & Average support size \\ \hline
        20NEWS  & 10\,989         & 85.553               \\
        AMAZON  & 15\,735         & 63.552               \\
        BBC     & 1\,780          & 190.133              \\ \hline
    \end{tabular}
    \label{tab:dataset}
\end{table}

As preprocessing, we tokenized and removed stop words, then embed tokens using the pretrained GloVe model~\cite{pennington2014glove}. 
To investigate the effect of embedding dimensionality on approximation performance, we tested dimensions of 50, 100, and 200. 
The methods evaluated for nearest‐neighbor search in this experiment were:
\begin{itemize}
    \item \textbf{Exact}: Computes the exact Wasserstein distance without approximation.
    \item \textbf{Sinkhorn}~\cite{cuturi2013sinkhorn}: A regularized OT approximation with with regularization parameter set to $0.1$ for $D=50,100$ and to $0.01$ for $D=200$, and a maximum of 10 iterations.
    \item \textbf{Flowtree}~\cite{backurs2020scalable}: Constructs a quadtree and computes the approximation cost using Algorithm~\ref{alg:flow}.
    \item \textbf{kd-Flowtree}: The proposed method, which constructs a kd‐tree and computes the approximation cost using Algorithm~\ref{alg:flow}.
    \item \textbf{$1$-Greedy}~\cite{houry2024fast}: The greedy approach $k$-Greedy with $k=1$. We evaluate both embedding schemes proposed in the paper: clustering‐tree embedding ($1$-Greedy(C)) and random-projection embedding ($1$-Greedy(R)).
\end{itemize}

For Exact and Sinkhorn, we used the Python library POT~\cite{flamary2021pot}. 
Other methods were implemented in C++, with evaluation code written in Python. 
We used the $\ell_1$ norm in the ground space, treating both query and target distributions as uniform (each support point has probability $1/s$ for support size $s$).   
Methods were evaluated by Recall@$k$ and total runtime. Recall@$k$ was measured as $k$ varies to assess approximation accuracy. 
Runtimes included both the total time for nearest‐neighbor search (including preprocessing) and the time excluding preprocessing, where preprocessing refers to tree construction or one-dimensional embedding.

The second experiment evaluated nearest‐neighbor search using Algorithm~\ref{alg:flow} while varying the maximum allowed tree depth to investigate its effect on approximation accuracy. 
We used the 20NEWS dataset embedded in 50 dimensions and perform searches with both Flowtree and kd‐Flowtree. 
In our implementation, a very large integer was set as the depth limit to prevent infinite recursion; here, we imposed smaller limits. 
When the limit was reached, all remaining points in the current cell became leaves, and tree construction stopped. We tested depth limits of 2, 5, 8, 10, 20 and 50 and observe how Recall@$k$ changes.

\subsection{Results of Approximate Nearest‐Neighbor Search}
Figure~\ref{fig:exp2-20news}, Figure~\ref{fig:exp2-amazon}, Figure~\ref{fig:exp2-bbc}, and Table~\ref{tab:exp2-time} present the nearest‐neighbor search results for each approximation method. 
Each figure shows results for different embedding dimensions on each dataset. 
The vertical axis represents the ratio of $k$ to the dataset size (e.g., for BBC at $r=0.1$, $k=1780\times0.1$), and the horizontal axis shows Recall@$k$. 
Table~\ref{tab:exp2-time} reports the runtimes for each dataset. For each method, the upper row gives the total runtime including preprocessing, and the lower row gives the runtime excluding preprocessing. 
Note that Exact and Sinkhorn have no preprocessing, so no values are provided in the lower row.  

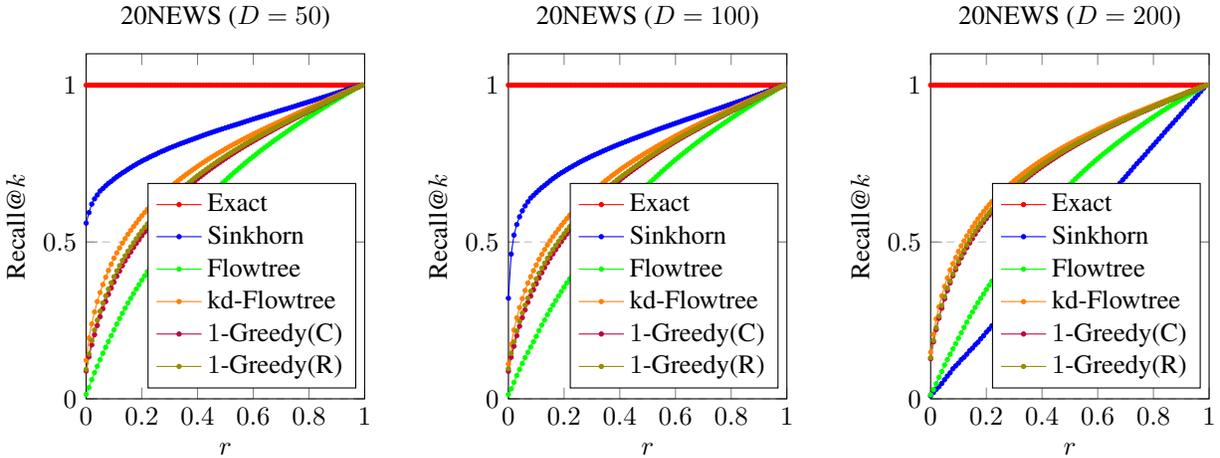
\begin{figure}[htbp]
    \centering
    \begin{subfigure}{0.32\textwidth}
        \centering
        \begin{tikzpicture}
            \begin{axis}[
                width=\linewidth,
                height=0.27\textheight,
                xlabel={$r$},
                ylabel={Recall@$k$},
                title={20NEWS ($D=50$)},
                xmin=0, xmax=1,
                ymin=0,
                xtick={0,0.2,0.4,0.6,0.8,1},
                ymajorgrids=true,
                grid style=dashed,
                legend pos=south east,
                legend cell align=left,
                legend columns=1
              ]
                \addplot+[] table [col sep=comma, header=true, x expr=\coordindex*0.01, y=Exact_mean] {Acr_20news_50d.csv};
                \addlegendentry{Exact}
                \addplot+[] table [col sep=comma, header=true, x expr=\coordindex*0.01, y=Sinkhorn_mean] {Acr_20news_50d.csv};
                \addlegendentry{Sinkhorn}
                \addplot+[] table [col sep=comma, header=true, x expr=\coordindex*0.01, y=FTQT_mean] {Acr_20news_50d.csv};
                \addlegendentry{Flowtree}
                \addplot+[] table [col sep=comma, header=true, x expr=\coordindex*0.01, y=FTKD_mean] {Acr_20news_50d.csv};
                \addlegendentry{kd-Flowtree}
                \addplot+[] table [col sep=comma, header=true, x expr=\coordindex*0.01, y=CLKG_mean] {Acr_20news_50d.csv};
                \addlegendentry{$1$-Greedy(C)}
                \addplot+[] table [col sep=comma, header=true, x expr=\coordindex*0.01, y=SLKG_mean] {Acr_20news_50d.csv};
                \addlegendentry{$1$-Greedy(R)}
            \end{axis}
        \end{tikzpicture}
    \end{subfigure}
    \hfill
    \begin{subfigure}{0.32\textwidth}
        \centering
        \begin{tikzpicture}
            \begin{axis}[
                width=\linewidth,
                height=0.27\textheight,
                xlabel={$r$},
                ylabel={Recall@$k$},
                title={20NEWS ($D=100$)},
                xmin=0, xmax=1,
                ymin=0,
                xtick={0,0.2,0.4,0.6,0.8,1},
                ymajorgrids=true,
                grid style=dashed,
                legend pos=south east,
                legend cell align=left,
                legend columns=1
              ]
                \addplot+[] table [col sep=comma, header=true, x expr=\coordindex*0.01, y=Exact_mean] {Acr_20news_100d.csv};
                \addlegendentry{Exact}
                \addplot+[] table [col sep=comma, header=true, x expr=\coordindex*0.01, y=Sinkhorn_mean] {Acr_20news_100d.csv};
                \addlegendentry{Sinkhorn}
                \addplot+[] table [col sep=comma, header=true, x expr=\coordindex*0.01, y=FTQT_mean] {Acr_20news_100d.csv};
                \addlegendentry{Flowtree}
                \addplot+[] table [col sep=comma, header=true, x expr=\coordindex*0.01, y=FTKD_mean] {Acr_20news_100d.csv};
                \addlegendentry{kd-Flowtree}
                \addplot+[] table [col sep=comma, header=true, x expr=\coordindex*0.01, y=CLKG_mean] {Acr_20news_100d.csv};
                \addlegendentry{$1$-Greedy(C)}
                \addplot+[] table [col sep=comma, header=true, x expr=\coordindex*0.01, y=SLKG_mean] {Acr_20news_100d.csv};
                \addlegendentry{$1$-Greedy(R)}
            \end{axis}
        \end{tikzpicture}
    \end{subfigure}
    \hfill
    \begin{subfigure}{0.32\textwidth}
        \centering
        \begin{tikzpicture}
            \begin{axis}[
                width=\linewidth,
                height=0.27\textheight,
                xlabel={$r$},
                ylabel={Recall@$k$},
                title={20NEWS ($D=200$)},
                xmin=0, xmax=1,
                ymin=0,
                xtick={0,0.2,0.4,0.6,0.8,1},
                ymajorgrids=true,
                grid style=dashed,
                legend pos=south east,
                legend cell align=left,
                legend columns=1
              ]
                \addplot+[] table [col sep=comma, header=true, x expr=\coordindex*0.01, y=Exact_mean] {Acr_20news_200d.csv};
                \addlegendentry{Exact}
                \addplot+[] table [col sep=comma, header=true, x expr=\coordindex*0.01, y=Sinkhorn_mean] {Acr_20news_200d.csv};
                \addlegendentry{Sinkhorn}
                \addplot+[] table [col sep=comma, header=true, x expr=\coordindex*0.01, y=FTQT_mean] {Acr_20news_200d.csv};
                \addlegendentry{Flowtree}
                \addplot+[] table [col sep=comma, header=true, x expr=\coordindex*0.01, y=FTKD_mean] {Acr_20news_200d.csv};
                \addlegendentry{kd-Flowtree}
                \addplot+[] table [col sep=comma, header=true, x expr=\coordindex*0.01, y=CLKG_mean] {Acr_20news_200d.csv};
                \addlegendentry{$1$-Greedy(C)}
                \addplot+[] table [col sep=comma, header=true, x expr=\coordindex*0.01, y=SLKG_mean] {Acr_20news_200d.csv};
                \addlegendentry{$1$-Greedy(R)}
            \end{axis}
        \end{tikzpicture}
    \end{subfigure}
    \caption{Recall@$k$ for 20NEWS}
    \label{fig:exp2-20news}
\end{figure}

\begin{figure}[htbp]
    \centering
    \begin{subfigure}{0.32\textwidth}
        \centering
        \begin{tikzpicture}
            \begin{axis}[
                width=\linewidth,
                height=0.27\textheight,
                xlabel={$r$},
                ylabel={Recall@$k$},
                title={AMAZON ($D=50$)},
                xmin=0, xmax=1,
                ymin=0,
                xtick={0,0.2,0.4,0.6,0.8,1},
                ymajorgrids=true,
                grid style=dashed,
                legend pos=south east,
                legend cell align=left,
                legend columns=1
              ]
                \addplot+[] table [col sep=comma, header=true, x expr=\coordindex*0.01, y=Exact_mean] {Acr_amazon_50d.csv};
                \addlegendentry{Exact}
                \addplot+[] table [col sep=comma, header=true, x expr=\coordindex*0.01, y=Sinkhorn_mean] {Acr_amazon_50d.csv};
                \addlegendentry{Sinkhorn}
                \addplot+[] table [col sep=comma, header=true, x expr=\coordindex*0.01, y=FTQT_mean] {Acr_amazon_50d.csv};
                \addlegendentry{Flowtree}
                \addplot+[] table [col sep=comma, header=true, x expr=\coordindex*0.01, y=FTKD_mean] {Acr_amazon_50d.csv};
                \addlegendentry{kd-Flowtree}
                \addplot+[] table [col sep=comma, header=true, x expr=\coordindex*0.01, y=CLKG_mean] {Acr_amazon_50d.csv};
                \addlegendentry{$1$-Greedy(C)}
                \addplot+[] table [col sep=comma, header=true, x expr=\coordindex*0.01, y=SLKG_mean] {Acr_amazon_50d.csv};
                \addlegendentry{$1$-Greedy(R)}
            \end{axis}
        \end{tikzpicture}
    \end{subfigure}
    \hfill
    \begin{subfigure}{0.32\textwidth}
        \centering
        \begin{tikzpicture}
            \begin{axis}[
                width=\linewidth,
                height=0.27\textheight,
                xlabel={$r$},
                ylabel={Recall@$k$},
                title={AMAZON ($D=100$)},
                xmin=0, xmax=1,
                ymin=0,
                xtick={0,0.2,0.4,0.6,0.8,1},
                ymajorgrids=true,
                grid style=dashed,
                legend pos=south east,
                legend cell align=left,
                legend columns=1
              ]
                \addplot+[] table [col sep=comma, header=true, x expr=\coordindex*0.01, y=Exact_mean] {Acr_amazon_100d.csv};
                \addlegendentry{Exact}
                \addplot+[] table [col sep=comma, header=true, x expr=\coordindex*0.01, y=Sinkhorn_mean] {Acr_amazon_100d.csv};
                \addlegendentry{Sinkhorn}
                \addplot+[] table [col sep=comma, header=true, x expr=\coordindex*0.01, y=FTQT_mean] {Acr_amazon_100d.csv};
                \addlegendentry{Flowtree}
                \addplot+[] table [col sep=comma, header=true, x expr=\coordindex*0.01, y=FTKD_mean] {Acr_amazon_100d.csv};
                \addlegendentry{kd-Flowtree}
                \addplot+[] table [col sep=comma, header=true, x expr=\coordindex*0.01, y=CLKG_mean] {Acr_amazon_100d.csv};
                \addlegendentry{$1$-Greedy(C)}
                \addplot+[] table [col sep=comma, header=true, x expr=\coordindex*0.01, y=SLKG_mean] {Acr_amazon_100d.csv};
                \addlegendentry{$1$-Greedy(R)}
            \end{axis}
        \end{tikzpicture}
    \end{subfigure}
    \hfill
    \begin{subfigure}{0.32\textwidth}
        \centering
        \begin{tikzpicture}
            \begin{axis}[
                width=\linewidth,
                height=0.27\textheight,
                xlabel={$r$},
                ylabel={Recall@$k$},
                title={AMAZON ($D=200$)},
                xmin=0, xmax=1,
                ymin=0,
                xtick={0,0.2,0.4,0.6,0.8,1},
                ymajorgrids=true,
                grid style=dashed,
                legend pos=south east,
                legend cell align=left,
                legend columns=1
              ]
                \addplot+[] table [col sep=comma, header=true, x expr=\coordindex*0.01, y=Exact_mean] {Acr_amazon_200d.csv};
                \addlegendentry{Exact}
                \addplot+[] table [col sep=comma, header=true, x expr=\coordindex*0.01, y=Sinkhorn_mean] {Acr_amazon_200d.csv};
                \addlegendentry{Sinkhorn}
                \addplot+[] table [col sep=comma, header=true, x expr=\coordindex*0.01, y=FTQT_mean] {Acr_amazon_200d.csv};
                \addlegendentry{Flowtree}
                \addplot+[] table [col sep=comma, header=true, x expr=\coordindex*0.01, y=FTKD_mean] {Acr_amazon_200d.csv};
                \addlegendentry{kd-Flowtree}
                \addplot+[] table [col sep=comma, header=true, x expr=\coordindex*0.01, y=CLKG_mean] {Acr_amazon_200d.csv};
                \addlegendentry{$1$-Greedy(C)}
                \addplot+[] table [col sep=comma, header=true, x expr=\coordindex*0.01, y=SLKG_mean] {Acr_amazon_200d.csv};
                \addlegendentry{$1$-Greedy(R)}
            \end{axis}
        \end{tikzpicture}
    \end{subfigure}
    \caption{Recall@$k$ for AMAZON}
    \label{fig:exp2-amazon}
\end{figure}
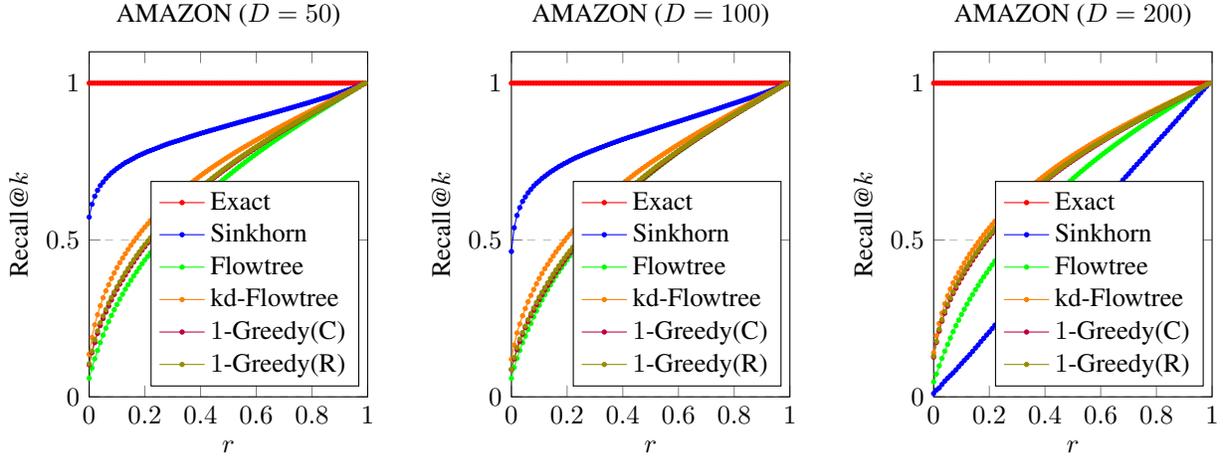

\begin{figure}[htbp]
    \centering
    \begin{subfigure}{0.32\textwidth}
        \centering
        \begin{tikzpicture}
            \begin{axis}[
                width=\linewidth,
                height=0.27\textheight,
                xlabel={$r$},
                ylabel={Recall@$k$},
                title={BBC ($D=50$)},
                xmin=0, xmax=1,
                ymin=0,
                xtick={0,0.2,0.4,0.6,0.8,1},
                ymajorgrids=true,
                grid style=dashed,
                legend pos=south east,
                legend cell align=left,
                legend columns=1
              ]
                \addplot+[] table [col sep=comma, header=true, x expr=\coordindex*0.01, y=Exact_mean] {Acr_bbc_50d.csv};
                \addlegendentry{Exact}
                \addplot+[] table [col sep=comma, header=true, x expr=\coordindex*0.01, y=Sinkhorn_mean] {Acr_bbc_50d.csv};
                \addlegendentry{Sinkhorn}
                \addplot+[] table [col sep=comma, header=true, x expr=\coordindex*0.01, y=FTQT_mean] {Acr_bbc_50d.csv};
                \addlegendentry{Flowtree}
                \addplot+[] table [col sep=comma, header=true, x expr=\coordindex*0.01, y=FTKD_mean] {Acr_bbc_50d.csv};
                \addlegendentry{kd-Flowtree}
                \addplot+[] table [col sep=comma, header=true, x expr=\coordindex*0.01, y=CLKG_mean] {Acr_bbc_50d.csv};
                \addlegendentry{$1$-Greedy(C)}
                \addplot+[] table [col sep=comma, header=true, x expr=\coordindex*0.01, y=SLKG_mean] {Acr_bbc_50d.csv};
                \addlegendentry{$1$-Greedy(R)}
            \end{axis}
        \end{tikzpicture}
    \end{subfigure}
    \hfill
    \begin{subfigure}{0.32\textwidth}
        \centering
        \begin{tikzpicture}
            \begin{axis}[
                width=\linewidth,
                height=0.27\textheight,
                xlabel={$r$},
                ylabel={Recall@$k$},
                title={BBC ($D=100$)},
                xmin=0, xmax=1,
                ymin=0,
                xtick={0,0.2,0.4,0.6,0.8,1},
                ymajorgrids=true,
                grid style=dashed,
                legend pos=south east,
                legend cell align=left,
                legend columns=1
              ]
                \addplot+[] table [col sep=comma, header=true, x expr=\coordindex*0.01, y=Exact_mean] {Acr_bbc_100d.csv};
                \addlegendentry{Exact}
                \addplot+[] table [col sep=comma, header=true, x expr=\coordindex*0.01, y=Sinkhorn_mean] {Acr_bbc_100d.csv};
                \addlegendentry{Sinkhorn}
                \addplot+[] table [col sep=comma, header=true, x expr=\coordindex*0.01, y=FTQT_mean] {Acr_bbc_100d.csv};
                \addlegendentry{Flowtree}
                \addplot+[] table [col sep=comma, header=true, x expr=\coordindex*0.01, y=FTKD_mean] {Acr_bbc_100d.csv};
                \addlegendentry{kd-Flowtree}
                \addplot+[] table [col sep=comma, header=true, x expr=\coordindex*0.01, y=CLKG_mean] {Acr_bbc_100d.csv};
                \addlegendentry{$1$-Greedy(C)}
                \addplot+[] table [col sep=comma, header=true, x expr=\coordindex*0.01, y=SLKG_mean] {Acr_bbc_100d.csv};
                \addlegendentry{$1$-Greedy(R)}
            \end{axis}
        \end{tikzpicture}
    \end{subfigure}
    \hfill
    \begin{subfigure}{0.32\textwidth}
        \centering
        \begin{tikzpicture}
            \begin{axis}[
                width=\linewidth,
                height=0.27\textheight,
                xlabel={$r$},
                ylabel={Recall@$k$},
                title={BBC ($D=200$)},
                xmin=0, xmax=1,
                ymin=0,
                xtick={0,0.2,0.4,0.6,0.8,1},
                ymajorgrids=true,
                grid style=dashed,
                legend pos=south east,
                legend cell align=left,
                legend columns=1
              ]
                \addplot+[] table [col sep=comma, header=true, x expr=\coordindex*0.01, y=Exact_mean] {Acr_bbc_200d.csv};
                \addlegendentry{Exact}
                \addplot+[] table [col sep=comma, header=true, x expr=\coordindex*0.01, y=Sinkhorn_mean] {Acr_bbc_200d.csv};
                \addlegendentry{Sinkhorn}
                \addplot+[] table [col sep=comma, header=true, x expr=\coordindex*0.01, y=FTQT_mean] {Acr_bbc_200d.csv};
                \addlegendentry{Flowtree}
                \addplot+[] table [col sep=comma, header=true, x expr=\coordindex*0.01, y=FTKD_mean] {Acr_bbc_200d.csv};
                \addlegendentry{kd-Flowtree}
                \addplot+[] table [col sep=comma, header=true, x expr=\coordindex*0.01, y=CLKG_mean] {Acr_bbc_200d.csv};
                \addlegendentry{$1$-Greedy(C)}
                \addplot+[] table [col sep=comma, header=true, x expr=\coordindex*0.01, y=SLKG_mean] {Acr_bbc_200d.csv};
                \addlegendentry{$1$-Greedy(R)}
            \end{axis}
        \end{tikzpicture}
    \end{subfigure}
    \caption{Recall@$k$ for BBC}
    \label{fig:exp2-bbc}
\end{figure}
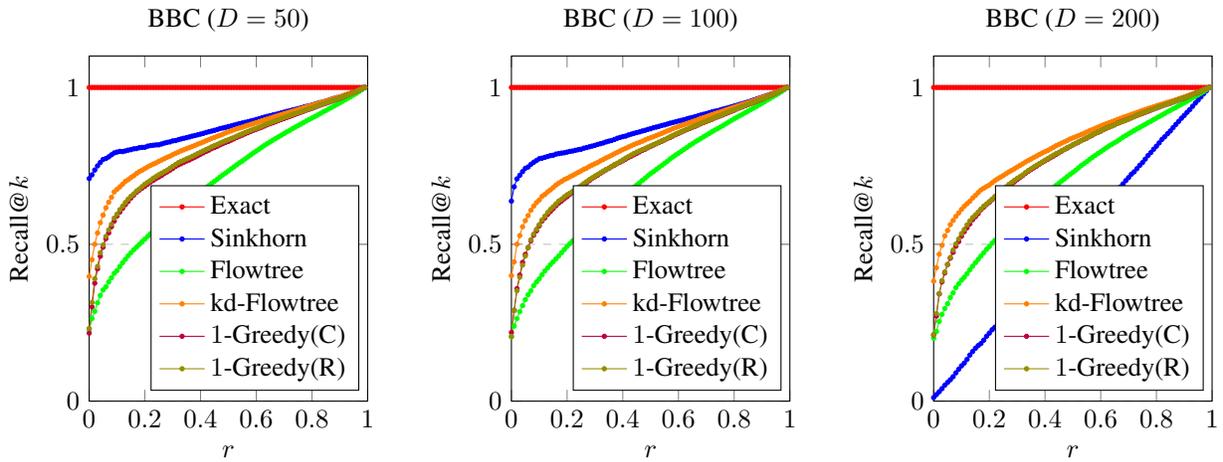

\begin{table}[htbp]
  \centering
  \caption{Execution Time (upper row: total runtime; lower row: time excluding preprocessing)}
  \begin{subtable}{\columnwidth}
    \centering
    \caption{Runtimes for $D=50$}
    \begin{tabular}{cccc} \hline
        Method & 20NEWS & AMAZON & BBC \\ \hline
        Exact &  
        \begin{tabular}{c}7.539\,s\\---\end{tabular} & 
        \begin{tabular}{c}6.281\,s\\---\end{tabular} & 
        \begin{tabular}{c}5.728\,s\\---\end{tabular}\\
        Sinkhorn &  
        \begin{tabular}{c}3.517\,s\\---\end{tabular} & 
        \begin{tabular}{c}4.140\,s\\---\end{tabular} & 
        \begin{tabular}{c}1.744\,s\\---\end{tabular}\\
        Flowtree &  
        \begin{tabular}{c}2.782\,s\\\textbf{0.946}\,s\end{tabular} & 
        \begin{tabular}{c}2.985\,s\\\textbf{0.849}\,s\end{tabular} & 
        \begin{tabular}{c}1.763\,s\\\textbf{0.938}\,s\end{tabular}\\
        kd-Flowtree &  
        \begin{tabular}{c}\textbf{2.335}\,s\\\textbf{0.948}\,s\end{tabular} & 
        \begin{tabular}{c}\textbf{2.235}\,s\\\textbf{0.832}\,s\end{tabular} & 
        \begin{tabular}{c}\textbf{1.570}\,s\\\textbf{0.942}\,s\end{tabular}\\
        $1$-Greedy(C) &  
        \begin{tabular}{c}4.504\,s\\0.993\,s\end{tabular} & 
        \begin{tabular}{c}4.048\,s\\0.892\,s\end{tabular} & 
        \begin{tabular}{c}2.425\,s\\0.976\,s\end{tabular}\\
        $1$-Greedy(R) &  
        \begin{tabular}{c}\textbf{1.028}\,s\\\textbf{0.920}\,s\end{tabular} & 
        \begin{tabular}{c}\textbf{0.941}\,s\\\textbf{0.821}\,s\end{tabular} & 
        \begin{tabular}{c}\textbf{0.980}\,s\\\textbf{0.940}\,s\end{tabular}\\\hline
    \end{tabular}
  \end{subtable}\\
  \vspace{1.5em}
  \begin{subtable}{\columnwidth}
    \centering
    \caption{Runtimes for $D=100$}
    \begin{tabular}{cccc} \hline
        Method & 20NEWS & AMAZON & BBC \\ \hline
        Exact &  
        \begin{tabular}{c}8.843\,s\\---\end{tabular} & 
        \begin{tabular}{c}9.349\,s\\---\end{tabular} & 
        \begin{tabular}{c}6.623\,s\\---\end{tabular}\\
        Sinkhorn &  
        \begin{tabular}{c}4.795\,s\\---\end{tabular} & 
        \begin{tabular}{c}5.966\,s\\---\end{tabular} & 
        \begin{tabular}{c}2.957\,s\\---\end{tabular}\\
        Flowtree &  
        \begin{tabular}{c}4.470\,s\\\textbf{1.754}\,s\end{tabular} & 
        \begin{tabular}{c}5.592\,s\\\textbf{2.004}\,s\end{tabular} & 
        \begin{tabular}{c}2.834\,s\\\textbf{1.691}\,s\end{tabular}\\
        kd-Flowtree &  
        \begin{tabular}{c}\textbf{3.156}\,s\\\textbf{1.755}\,s\end{tabular} & 
        \begin{tabular}{c}\textbf{3.626}\,s\\\textbf{1.987}\,s\end{tabular} & 
        \begin{tabular}{c}\textbf{2.286}\,s\\\textbf{1.694}\,s\end{tabular}\\
        $1$-Greedy(C) &  
        \begin{tabular}{c}8.039\,s\\1.780\,s\end{tabular} & 
        \begin{tabular}{c}8.512\,s\\2.057\,s\end{tabular} & 
        \begin{tabular}{c}4.191\,s\\1.726\,s\end{tabular}\\
        $1$-Greedy(R) &  
        \begin{tabular}{c}\textbf{1.859}\,s\\\textbf{1.729}\,s\end{tabular} & 
        \begin{tabular}{c}\textbf{2.128}\,s\\\textbf{1.972}\,s\end{tabular} & 
        \begin{tabular}{c}\textbf{1.731}\,s\\\textbf{1.687}\,s\end{tabular}\\\hline
    \end{tabular}
  \end{subtable}\\
  \vspace{1.5em}
  \begin{subtable}{\columnwidth}
    \centering
    \caption{Runtimes for $D=200$}
    \begin{tabular}{cccc} \hline
        Method & 20NEWS & AMAZON & BBC \\ \hline
        Exact &  
        \begin{tabular}{c}11.579\,s\\---\end{tabular} & 
        \begin{tabular}{c}9.787\,s\\---\end{tabular} & 
        \begin{tabular}{c}9.245\,s\\---\end{tabular}\\
        Sinkhorn &  
        \begin{tabular}{c}5.730\,s\\---\end{tabular} & 
        \begin{tabular}{c}5.248\,s\\---\end{tabular} & 
        \begin{tabular}{c}5.339\,s\\---\end{tabular}\\
        Flowtree &  
        \begin{tabular}{c}7.717\,s\\\textbf{3.331}\,s\end{tabular} & 
        \begin{tabular}{c}8.173\,s\\\textbf{2.978}\,s\end{tabular} & 
        \begin{tabular}{c}5.163\,s\\\textbf{3.330}\,s\end{tabular}\\
        kd-Flowtree &  
        \begin{tabular}{c}\textbf{4.725}\,s\\\textbf{3.338}\,s\end{tabular} & 
        \begin{tabular}{c}\textbf{4.375}\,s\\\textbf{2.958}\,s\end{tabular} & 
        \begin{tabular}{c}\textbf{3.924}\,s\\\textbf{3.330}\,s\end{tabular}\\
        $1$-Greedy(C) &  
        \begin{tabular}{c}19.372\,s\\3.384\,s\end{tabular} & 
        \begin{tabular}{c}16.929\,s\\3.040\,s\end{tabular} & 
        \begin{tabular}{c}9.738\,s\\3.363\,s\end{tabular}\\
        $1$-Greedy(R) &  
        \begin{tabular}{c}\textbf{3.466}\,s\\\textbf{3.305}\,s\end{tabular} & 
        \begin{tabular}{c}\textbf{3.131}\,s\\\textbf{2.945}\,s\end{tabular} & 
        \begin{tabular}{c}\textbf{3.368}\,s\\\textbf{3.314}\,s\end{tabular}\\\hline
    \end{tabular}
  \end{subtable}
  \label{tab:exp2-time}
\end{table}

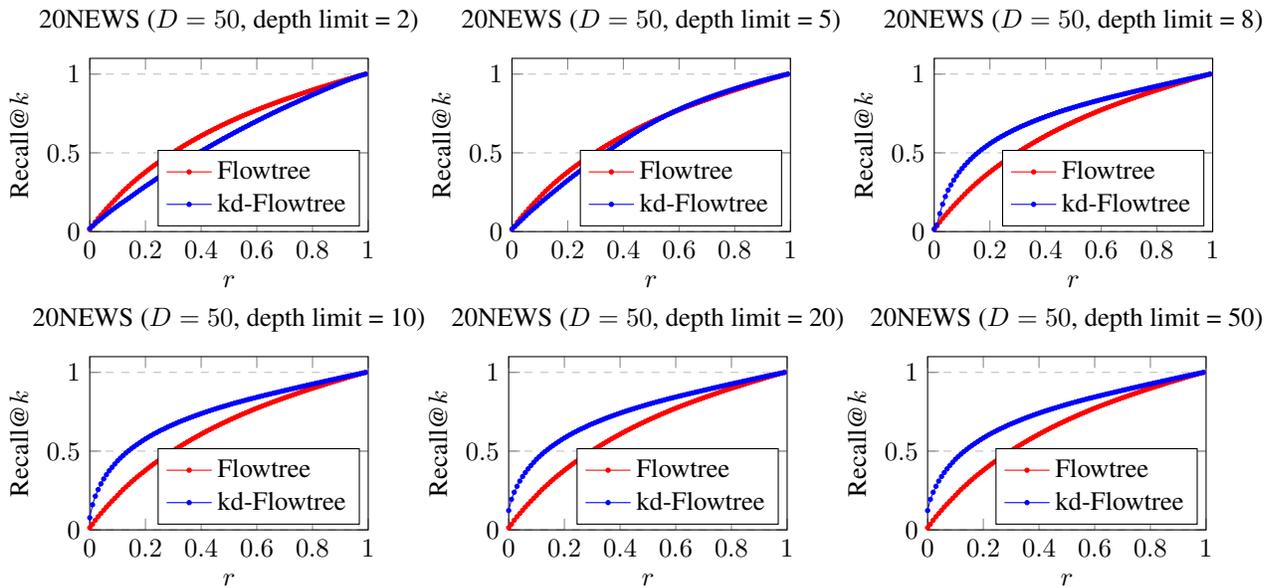
\begin{figure}[htbp]
    \centering
    \begin{subfigure}{0.32\textwidth}
        \centering
        \begin{tikzpicture}
            \begin{axis}[
                width=\columnwidth,
                height=0.17\textheight,
                xlabel={$r$},
                ylabel={Recall@$k$},
                title={20NEWS ($D=50$, depth limit = 2)},
                xmin=0, xmax=1,
                ymin=0,
                xtick={0,0.2,0.4,0.6,0.8,1},
                ymajorgrids=true,
                grid style=dashed,
                legend pos=south east,
                legend cell align=left,
                legend columns=1
              ]
                \addplot+[] table [col sep=comma, header=true, x expr=\coordindex*0.01, y=FTQT_mean] {Acr_50d_max2.csv};
                \addlegendentry{Flowtree}
                \addplot+[] table [col sep=comma, header=true, x expr=\coordindex*0.01, y=FTKD_mean] {Acr_50d_max2.csv};
                \addlegendentry{kd-Flowtree}
            \end{axis}
        \end{tikzpicture}
    \end{subfigure}
    \hfill
    \begin{subfigure}{0.32\textwidth}
        \centering
        \begin{tikzpicture}
            \begin{axis}[
                width=\columnwidth,
                height=0.17\textheight,
                xlabel={$r$},
                ylabel={Recall@$k$},
                title={20NEWS ($D=50$, depth limit = 5)},
                xmin=0, xmax=1,
                ymin=0,
                xtick={0,0.2,0.4,0.6,0.8,1},
                ymajorgrids=true,
                grid style=dashed,
                legend pos=south east,
                legend cell align=left,
                legend columns=1
              ]
                \addplot+[] table [col sep=comma, header=true, x expr=\coordindex*0.01, y=FTQT_mean] {Acr_50d_max5.csv};
                \addlegendentry{Flowtree}
                \addplot+[] table [col sep=comma, header=true, x expr=\coordindex*0.01, y=FTKD_mean] {Acr_50d_max5.csv};
                \addlegendentry{kd-Flowtree}
            \end{axis}
        \end{tikzpicture}
    \end{subfigure}
    \hfill
    \begin{subfigure}{0.32\textwidth}
        \centering
        \begin{tikzpicture}
            \begin{axis}[
                width=\columnwidth,
                height=0.17\textheight,
                xlabel={$r$},
                ylabel={Recall@$k$},
                title={20NEWS ($D=50$, depth limit = 8)},
                xmin=0, xmax=1,
                ymin=0,
                xtick={0,0.2,0.4,0.6,0.8,1},
                ymajorgrids=true,
                grid style=dashed,
                legend pos=south east,
                legend cell align=left,
                legend columns=1
              ]
                \addplot+[] table [col sep=comma, header=true, x expr=\coordindex*0.01, y=FTQT_mean] {Acr_50d_max8.csv};
                \addlegendentry{Flowtree}
                \addplot+[] table [col sep=comma, header=true, x expr=\coordindex*0.01, y=FTKD_mean] {Acr_50d_max8.csv};
                \addlegendentry{kd-Flowtree}
            \end{axis}
        \end{tikzpicture}
    \end{subfigure}
    \vspace{2em}
    \begin{subfigure}{0.32\textwidth}
        \centering
        \begin{tikzpicture}
            \begin{axis}[
                width=\columnwidth,
                height=0.17\textheight,
                xlabel={$r$},
                ylabel={Recall@$k$},
                title={20NEWS ($D=50$, depth limit = 10)},
                xmin=0, xmax=1,
                ymin=0,
                xtick={0,0.2,0.4,0.6,0.8,1},
                ymajorgrids=true,
                grid style=dashed,
                legend pos=south east,
                legend cell align=left,
                legend columns=1
              ]
                \addplot+[] table [col sep=comma, header=true, x expr=\coordindex*0.01, y=FTQT_mean] {Acr_50d_max10.csv};
                \addlegendentry{Flowtree}
                \addplot+[] table [col sep=comma, header=true, x expr=\coordindex*0.01, y=FTKD_mean] {Acr_50d_max10.csv};
                \addlegendentry{kd-Flowtree}
            \end{axis}
        \end{tikzpicture}
    \end{subfigure}
    \hfill
    \begin{subfigure}{0.32\textwidth}
        \centering
        \begin{tikzpicture}
            \begin{axis}[
                width=\columnwidth,
                height=0.17\textheight,
                xlabel={$r$},
                ylabel={Recall@$k$},
                title={20NEWS ($D=50$, depth limit = 20)},
                xmin=0, xmax=1,
                ymin=0,
                xtick={0,0.2,0.4,0.6,0.8,1},
                ymajorgrids=true,
                grid style=dashed,
                legend pos=south east,
                legend cell align=left,
                legend columns=1
              ]
                \addplot+[] table [col sep=comma, header=true, x expr=\coordindex*0.01, y=FTQT_mean] {Acr_50d_max20.csv};
                \addlegendentry{Flowtree}
                \addplot+[] table [col sep=comma, header=true, x expr=\coordindex*0.01, y=FTKD_mean] {Acr_50d_max20.csv};
                \addlegendentry{kd-Flowtree}
            \end{axis}
        \end{tikzpicture}
    \end{subfigure}
    \hfill
    \begin{subfigure}{0.32\textwidth}
        \centering
        \begin{tikzpicture}
            \begin{axis}[
                width=\columnwidth,
                height=0.17\textheight,
                xlabel={$r$},
                ylabel={Recall@$k$},
                title={20NEWS ($D=50$, depth limit = 50)},
                xmin=0, xmax=1,
                ymin=0,
                xtick={0,0.2,0.4,0.6,0.8,1},
                ymajorgrids=true,
                grid style=dashed,
                legend pos=south east,
                legend cell align=left,
                legend columns=1
              ]
                \addplot+[] table [col sep=comma, header=true, x expr=\coordindex*0.01, y=FTQT_mean] {Acr_50d_max50.csv};
                \addlegendentry{Flowtree}
                \addplot+[] table [col sep=comma, header=true, x expr=\coordindex*0.01, y=FTKD_mean] {Acr_50d_max50.csv};
                \addlegendentry{kd-Flowtree}
            \end{axis}
        \end{tikzpicture}
    \end{subfigure}
    \vspace{-20pt}
    \caption{Recall@$k$ change with different depth limits}
    \label{fig:exp1-5d}
\end{figure}

First, looking at Recall@$k$, Sinkhorn achieves the highest accuracy across all datasets, followed by the proposed kd-Flowtree. 
For 20NEWS, when $D$ is large, $1$-Greedy approaches kd-Flowtree’s performance; however, for AMAZON and BBC, kd-Flowtree consistently outperforms the other methods. 
Thus, kd-Flowtree provides stable and high approximation accuracy regardless of embedding dimensionality.  
Although the previous section indicated that search accuracy depends on $D$, in practice, when $D=200$, accuracy slightly degrades compared to $D=50$ or $100$, narrowing the gap with existing methods. 
Nevertheless, kd-Flowtree still maintains superior performance.

Next, examining runtimes excluding preprocessing, Flowtree, kd-Flowtree, and $1$-Greedy complete searches the fastest. 
As noted, kd-Flowtree may incur slightly longer query time than Flowtree due to deeper trees, but the difference is negligible. 
For total runtime including preprocessing, kd-Flowtree and $1$-Greedy(R) are the fastest. In particular, $1$-Greedy(R) shows almost no overhead between total and query‐only time, indicating low preprocessing cost. 
In contrast, $1$-Greedy(C) requires more time than even Exact, revealing that clustering‐tree construction is computationally expensive.  
Moreover, the gap in total runtime between kd-Flowtree and Flowtree widens as $D$ increases, highlighting the advantage of kd‐trees in high‐dimensional spaces.   

These results indicate that kd-Flowtree not only exceeds previous methods in approximation accuracy but also maintains competitive runtimes and low preprocessing cost. 
While $1$-Greedy(C) can rival kd-Flowtree in query speed, its high preprocessing cost makes it less practical for large‐scale datasets. 
Moreover, compared to $1$-Greedy(R), kd-Flowtree yields more stable and higher accuracy for comparable runtimes.  

\subsection{Effect of Tree Depth Limit on Nearest‐Neighbor Search Accuracy}
Figure~\ref{fig:exp1-5d} shows nearest‐neighbor search results using Flowtree and kd-Flowtree as we vary the maximum tree depth. 
Each subplot corresponds to a different depth limit; the vertical axis represents the ratio of $k$ to the dataset size, and the horizontal axis shows Recall@$k$.

As the depth limit increases, Flowtree’s approximation accuracy remains nearly constant, whereas kd-Flowtree’s accuracy improves. 
Since the average support size for 20NEWS is about 85, a quadtree in 50 dimensions typically completes subdivision in a single split—except in rare cases of duplicate points. 
Therefore, increasing the depth limit does not change the quadtree structure or search accuracy. In contrast, kd‐trees continue to split deeply even in high dimensions. 
With a very low depth limit, construction is prematurely truncated and accuracy suffers; as the limit increases, tree construction proceeds appropriately and accuracy improves.

When the limit is 2, the kd‐tree barely grows, resulting in lower accuracy than Flowtree. 
At a limit of 5, the tree grows sufficiently to match Flowtree’s performance. 
Further increasing the limit allows kd-Flowtree to surpass Flowtree, but beyond a limit of 8, additional depth yields only marginal gains, suggesting that the tree is effectively complete. 
For example, if the number of points to split is roughly twice the average support size (around 190), then $2^8=256$ regions suffice to isolate all points by depth 8. 
Thus, in practice, setting the kd‐tree depth limit to slightly above twice the average support size is sufficient for high accuracy.

These results demonstrate that the performance of any Algorithm~\ref{alg:flow}–based method is highly sensitive to tree depth: insufficient depth degrades approximation accuracy. 
In particular, quadtrees remain very shallow in high‐dimensional spaces, while kd‐trees can grow deep, explaining kd-Flowtree’s superior accuracy observed in Experiment 1.  

\section{Conclusion}
In this paper, we proposed \emph{kd-Flowtree}, a kd‐tree–based approximation method for the Wasserstein distance.  
The accuracy of nearest‐neighbor search with the Flowtree algorithm improves as the depth of the constructed tree increases.  
Since kd‐trees grow deep even in high dimensions—unlike quadtrees—kd-Flowtree achieves superior search performance in high‐dimensional spaces.  
We also proved that kd-Flowtree’s search accuracy is independent of the dataset size.  

Our results suggest the effectiveness of leveraging kd‐trees in optimal transport research. 
Although kd‐tree split positions depend on data distribution—requiring careful consideration in theoretical analyses—the construction procedure is simple, runs quickly, and is well suited to high dimensions. 
We hope this idea will find application beyond nearest‐neighbor search in other domains of optimal transport.   

\bibliographystyle{unsrtnat}
\bibliography{reference}

\end{document}